\documentclass[11pt, letterpaper]{article}
\usepackage[margin=1in]{geometry}

\usepackage{amsfonts, amsmath, amssymb, amsthm}
\usepackage[nocompress]{cite}
\usepackage{color}
\usepackage{euscript}
\usepackage{graphicx}
\usepackage{microtype}
\usepackage{wrapfig}

\def\extraspacing{\vspace{4mm} \noindent}
\def\figcapup{\vspace{-1mm}}
\def\figcapdown{\vspace{-0mm}}

\def\vgap{\vspace{2mm}}

\def\tabpos{\hspace{4mm} \= \hspace{4mm} \= \hspace{4mm} \= \hspace{4mm} \=
\hspace{4mm} \= \hspace{4mm} \= \hspace{4mm} \= \hspace{4mm} \= \hspace{4mm}
\kill}
\newtheorem{theorem}{Theorem}
\newtheorem{lemma}{Lemma}

\newtheorem{proposition}{Proposition}

\newcommand{\bm}[1]{\boldmath{#1}}

\def\eps{\epsilon}
\def\fr{\frac}
\def\-{\mbox{-}}

\def\real{\mathbb{R}}

\def\lc{\lceil}

\def\rc{\rceil}

\def\nn{\nonumber}

\def\Pr{\mathbf{Pr}}

\def\*{\star}

\DeclareMathOperator*{\polylog}{polylog}

\def\figcapup{\vspace{0mm}}
\def\figcapdown{\vspace{0mm}}
\def\extraspacing{\vspace{3mm} \noindent}

\def\vgap{\vspace{2mm}}

\def\dis{\mathit{DIS}}
\def\ol{\overline}
\def\supp{\mathit{supp}}

\def\R{\mathcal{R}}
\def\D{\mathcal{D}}
\def\U{\mathcal{U}}
\def\X{\mathcal{X}}

\title{Distribution-Sensitive Bounds on \\Relative Approximations of Geometric Ranges}

\author{
        Yufei Tao \hspace{15mm} Yu Wang \\[2mm]
        \small{Chinese University of Hong Kong} \\
        \small{Hong Kong} \\ 
        {\small \{{\em taoyf}, {\em yuwang}\}{\em @cse.cuhk.edu.hk}} \\[5mm]
}

\date{} 

\begin{document}
\begin{sloppy}
\maketitle
\begin{abstract}

        A family $\mathcal{R}$ of ranges and a set $X$ of points, all in $\real^d$, together define a range space $(X, \mathcal{R}|_X)$, where $\mathcal{R}|_X = \{X \cap h \mid h \in \mathcal{R}\}$. We want to find a structure to estimate the quantity $|X \cap h|/|X|$ for any range $h \in \mathcal{R}$ with the {\em $(\rho, \epsilon)$-guarantee}: (i) if $|X \cap h|/|X| > \rho$, the estimate must have a relative error $\epsilon$; (ii) otherwise, the estimate must have an absolute error $\rho \epsilon$. The objective is to minimize the size of the structure. Currently, the dominant solution is to compute a relative $(\rho, \epsilon)$-approximation, which is a subset of $X$ with $\tilde{O}(\lambda/(\rho \epsilon^2))$ points, where $\lambda$ is the VC-dimension of $(X, \mathcal{R}|_X)$, and $\tilde{O}$ hides polylog factors.
        
        \vgap
        
        This paper shows a more general bound sensitive to the content of $X$. We give a structure that stores $O(\log (1/\rho))$ integers plus $\tilde{O}(\theta \cdot (\lambda/\epsilon^2))$ points of $X$, where $\theta$ --- called the {\em disagreement coefficient} --- measures how much the ranges differ from each other in their intersections with $X$. The value of $\theta$ is between 1 and $1/\rho$, such that our space bound is never worse than that of relative $(\rho, \epsilon)$-approximations, but we improve the latter's $1/\rho$ term whenever $\theta = o(\fr{1}{\rho \log (1/\rho)})$. We also prove that, in the worst case, summaries with the $(\rho, 1/2)$-guarantee must consume $\Omega(\theta)$ words even for $d = 2$ and $\lambda \le 3$.
        
        \vgap
        
        We then constrain $\R$ to be the set of halfspaces in $\real^d$ for a constant $d$, and prove the existence of structures with $o(1/(\rho \epsilon^2))$ size offering $(\rho,\epsilon)$-guarantees, when $X$ is generated from various stochastic distributions. This is the first formal justification on why the term $1/\rho$ is not compulsory for "realistic" inputs.
\end{abstract}

\thispagestyle{empty}
\clearpage
\setcounter{page}{1}
   
\vspace{5mm} 

\newpage    
\pagebreak

\section{Introduction} \label{sec::intro} 

A (data) {\em summary}, in general, refers to a structure that captures certain information up to a specified precision about a set of objects, but using space significantly smaller than the size of the set. These summaries have become important tools in algorithm design, especially in distributed/parallel computing where the main performance goal is to minimize the communication across different servers. 

\vgap

In this paper, we revisit the problem of finding a small-space summary to perform range estimation in $\real^d$ with relative-error guarantees. Let $\R$ be a family of geometric ranges in $\real^d$ (e.g., a ``halfspace family'' $\R$ is the set of all halfspaces in $\real^d$), and $X$ be a set of points in $\real^d$. $\R$ and $X$ together define a range space $(X, \R|_X)$, where $\R|_X = \{X \cap h \mid h \in \R\}$. Denote by  $\lambda$ the VC-dimension of $(X, \R|_X)$. 

\vgap

Following the notations of \cite{e13,hs11b}, define 
\begin{eqnarray} 
    \ol{X}(h) = |X \cap h|/|X| \nn
\end{eqnarray}
for each $h \in \R$, namely, $\ol{X}(h)$ is the fraction of points in $X$ that are covered by $h$. Given real-valued parameters $0 < \rho, \eps < 1$, we need to produce a structure --- called a {\em $(\rho, \eps)$-summary} henceforth --- that allows us to produce, for every range $h \in \R$, a real-valued estimate $\tau$ satisfying the following {\em $(\rho, \eps)$-guarantee}: 
\begin{eqnarray}
    \left| \ol{X}(h) - \tau \right| &\leq& \eps \cdot \max\{\rho, \ol{X}(h)\}. \label{eqn::intro-rho_eps_guarantee}
\end{eqnarray}
Phrased differently, the guarantee says that (i) if $\ol{X}(h) > \rho$, $\tau$ must have a relative error at most $\eps$; (ii) otherwise, $\tau$ must have an absolute error at most $\rho \eps$. The main challenge is to minimize the size of the structure. 

\subsection{Previous Results} \label{sec::intro-prev}

Throughout the paper, all logarithms have base 2 by default. 

\extraspacing {\bf Sample-Based \bm{$(\rho,\eps)$}-Summaries.} We say that a $(\rho, \eps)$-summary of $(X, \R|_X)$ is {\em sample-based} if it meets the requirements below: it stores a subset $Z \subseteq X$ such that, for any range $h \in \R$ with $Z \cap h = \emptyset$, it returns an estimate 0 for $\ol{X}(h)$. 

\vgap

A {\em relative $(\rho, \eps)$-approximation} \cite{hs11b,lls01} is a subset $Z \subseteq X$ such that $\left| \ol{X}(h) - \ol{Z}(h) \right| \leq \eps \cdot \max\{\rho, \ol{X}(h)\}$ holds for all ranges $h \in \R$. Hence, the $(\rho, \eps)$-guarantee can be fulfilled by simply setting $\tau$ to $\ol{Z}(h)$, rendering $Z$ a legal (sample-based) $(\rho, \eps)$-summary. Strengthening earlier results \cite{bcm99,h92,p86}, Li et al.\ \cite{lls01} proved that a random sample of $X$ with size $O(\frac{1}{\rho} \cdot \fr{1}{\eps^2} ( \lambda \log \fr{1}{\rho} + \log\fr{1}{\delta}))$ is a relative $(\rho, \eps)$-approximation with probability at least $1 - \delta$. This implies the existence of a $(\rho, \eps)$-summary of size $O(\frac{1}{\rho} \cdot \fr{\lambda }{\eps^2}  \log \fr{1}{\rho})$.

\vgap 

A range space $(X, \R|_X)$ of a constant VC-dimension is said to be {\em well-behaved}, if $\R|_X$ contains at most $O(|X|) \cdot k^{O(1)}$ sets of size not exceeding $k$, for any integer $k$ from $1$ to $|X|$. Ezra \cite{e13} showed that such a range space admits a sample-based $(\rho, \eps)$-summary of size $O(\frac{1}{\rho} \cdot \fr{1}{\eps^2} (\log \fr{1}{\eps} +  \log\log \fr{1}{\rho}))$; note that this is smaller than the corresponding result $O(\frac{1}{\rho} \cdot \fr{1}{\eps^2}  \log \fr{1}{\rho})$ of \cite{lls01} when $\rho \ll \eps$. It is worth mentioning that, when $d \le 3$ and $\R$ is the halfspace family, any $(X, \R|_X)$ is well-behaved; this, however, is not true when $d \ge 4$. 

\vgap 

As opposed to the above ``generic'' bounds, Har-Peled and Sharir \cite{hs11b} proved specific bounds on the halfspace family $\R$. For $d = 2$, they showed that any $(X, \R)$ has a relative $(\rho, \eps)$-approximation of size $O(\fr{1}{\rho} \cdot \fr{1}{\eps^{4/3}} \log^{4/3} \fr{1}{\rho \eps})$; similarly, for $d = 3$, the bound becomes $O(\fr{1}{\rho} \cdot \fr{1}{\eps^{3/2}} \log^{3/2} \fr{1}{\rho \eps})$. Combining these results and those of \cite{e13} gives the currently best bounds for these range spaces.

\extraspacing {\bf A Lower Bound of \bm{$\Omega(1/\rho)$}.} Notice that all the above bounds contain a term $1/\rho$. This is not a coincidence, but instead is due to a connection to ``$\eps$-nets''. Given a range space $(X, \R|_X)$, an {\em $\eps$-net} \cite{hw87} is a subset $Z \subseteq X$ such that $\ol{Z}(h) > 0$ holds for any range $h \in \R$ satisfying $\ol{X}(h) \ge \eps$. As can be verified easily, any sample-based $(\rho, 1/2)$-summary of $(X, \R|_X)$ must also be a $\rho$-net. This implies that the smallest size of sample-based $(\rho, 1/2)$-summaries must be at least that of $\rho$-nets.

\vgap

Regarding the sizes of $\eps$-nets, a lower bound of $\Omega(\fr{1}{\eps} \log \fr{1}{\eps})$ is known for many range families $\R$ (see \cite{pt11,kmp16,kpw92} and the references therein). More precisely, this means that, for each such family $\R$, one cannot hope to obtain an $\eps$-net of size $o(\fr{1}{\eps} \log \fr{1}{\eps})$ for {\em every} possible $X$. It thus follows that, for these families $\R$, $\Omega(\fr{1}{\rho} \log \fr{1}{\rho})$ is a lower bound on the sizes of sample-based $(\rho, 1/2)$-summaries. This, in turn, indicates that range spaces $(X, \R|_X)$ defined by such an $\R$ cannot always be well-behaved.

\vgap

Coming back to the sizes of $\eps$-nets, a weaker lower bound of $\Omega(1/\eps)$ holds quite commonly even on the range families that evade the $\Omega(\fr{1}{\eps} \log \fr{1}{\eps})$ bound. Consider, for example, the halfspace family $\R$ in $\real^2$. For any $X$, the range space $(X, \R|_X)$ definitely has an $\eps$-net of size $O(1/\eps)$ \cite{hkss14,msw90}. This is tight: place a set $X$ of points on the boundary of a circle; and it is easy to show that any $\eps$-net of $(X, \R|_X)$ must have a size of at least $1/\eps$. This means that the size of any sample-based $(\rho, 1/2)$-summary of $(X, \R|_X)$ must be $\Omega(1/\rho)$. 

\subsection{Our Results} \label{sec::intro-ours}

\extraspacing {\bf On One Input: Moving Beyond \bm{$\Omega(1/\rho)$}.} The $\Omega(1/\rho)$ lower bound discussed earlier holds only in the {\em worst case}, i.e., it is determined by the ``hardest'' $X$. For other $X$, the range space $(X, \R|_X)$ may admit much smaller $(\rho, \eps)$-summaries. For example, let $\R$ be again the set of halfplanes in $\real^2$. When all the points of $X$ lie on a line, $(X, \R)$ has a sample-based $(\rho, \eps)$-approximation of size only $O(\fr{1}{\eps} \log \fr{1}{\rho})$; also, it would be interesting to note that $(X, \R)$ has an $\eps$-net that contains only 2 points! In general, the existing bounds on $(\rho, \eps)$-summaries can be excessively loose on individual inputs $X$. This calls for an alternative, {\em distribution-sensitive}, analytical framework that is able to prove tighter bounds using extra complexity parameters that depend on the {\em content} of $X$. 

\vgap

The first contribution of this paper is to establish such a framework by resorting to the concept of {\em disagreement coefficient} \cite{h07} from active learning. This notion was originally defined in a context different from ours; and we will adapt it to range spaces in the next section. At this moment, it suffices to understand that the disagreement coefficient $\theta$ is a real value satisfying: $1 \le \theta \le 1/\rho$. The coefficient  quantifies the differences among the sets in $\R|_X$ (a larger $\theta$ indicates greater differences). Even under the same $\R$, $\theta$ may vary considerably depending on $X$. 

\vgap

We will show that, for any range space $(X, \R|_X)$ of VC-dimension $\lambda$, there is a $(\rho, \eps)$-summary that keeps $O(\log(1/\rho))$ integers plus 
\begin{eqnarray} 
    O\left(\min\Big\{\fr{1}{\rho}, \theta \log \fr{1}{\rho} \Big\} \cdot \fr{\lambda }{\eps^2}  \log \fr{1}{\rho} \right)
    \label{eqn::ours-finalbound}
\end{eqnarray}
points of $X$. The above is never worse than the general bound $O(\fr{1}{\rho} \cdot \fr{\lambda}{\eps^2}  \log \fr{1}{\rho})$ of relative $(\rho,\eps)$-approximations. 

\vgap

We will also prove that $\Omega(\theta)$ is a lower bound on the number of words needed to encode a $(\rho, 1/2)$-summary even when $d = 2$ and $\lambda \le 3$. This generalizes the $\Omega(1/\rho)$ lower bound in Section~\ref{sec::intro-prev} because $\theta$ is at most, but can reach, $1/\rho$. Thus, our result in \eqref{eqn::ours-finalbound} reflects the hardness of the input, and is tight within polylog factors for constant $\eps$. Our lower bound is information-theoretic, and does not require the summary to be sample-based. 

\extraspacing {\bf On a Distribution of Inputs: Small Summaries for Halfspaces.} Our framework allows us to explain --- for the first time we believe --- why $\Omega(1/\rho) \cdot \textrm{poly}(1/\eps)$ is too pessimistic a bound on the sizes of $(\rho,\eps)$-summaries for  inputs encountered in practice. For this purpose, we {\em must not} allow arbitrary inputs because of the prevalent $\Omega(1/\rho)$ lower bound; instead, we will examine a class of inputs following a certain distribution. 

\vgap

In this paper, we demonstrate the above by concentrating on the family $\R$ of halfspaces in $\real^d$ where the dimensionality $d$ is a constant; this is arguably the ``most-studied'' family in the literature of relative $(\rho,\eps)$-approximations. The core of our solutions concerns two stochastic distributions that have drastically different behavior:

\begin{itemize} 
    \item[$-$] {\bf (Box Uniform)} Suppose that $X$ is obtained by drawing $n$ points uniformly at random from the {\em unit box} $[0, 1]^d$. When $\rho = \Omega(\fr{\log n}{n})$, we will prove that $\theta = O(\polylog \fr{1}{\rho})$ with high probability (i.e., at least $1-1/n^2$). Accordingly,  \eqref{eqn::ours-finalbound} becomes $O(\fr{1}{\eps^2} \polylog \fr{1}{\rho})$, improving the general bound of relative $(\rho,\eps)$-approximations by almost a factor of $O(1/\rho)$. 
    
    \vgap
    
    \item[$-$] {\bf (Ball Uniform)} Consider instead that the $n$ points are drawn uniformly at random from the {\em unit ball}: $\{x \in \real^d \mid \sum_{i=1}^d x[i]^2 \leq 1\}$, where $x[i]$ represents the $i$-th coordinate of point $x$. This time, we will prove that $\theta = O((\fr{1}{\rho})^\fr{d-1}{d+1})$ with high probability for $\rho = \Omega(\fr{\log n}{n})$. \eqref{eqn::ours-finalbound} indicates the existence of a $(\rho, \eps)$-summary with size $\tilde{O}((\fr{1}{\rho})^\fr{d-1}{d+1} \cdot \fr{1}{\eps^2})$ for $\rho = \Omega(\fr{\log n}{n})$, again circumventing the $\Omega(1/\rho)$ lower bound.
\end{itemize}



The very same bounds can also be obtained in {\em non-uniform} settings. Suppose that $X$ is obtained by drawing $n$ points in an {\em iid} manner, according to a distribution that can be described by a probabilistic density function (pdf) $\pi(x)$ over $\real^d$ where $d = O(1)$. Define the {\em support region} of $\pi$ as $\supp(\pi) = \{x \in \real^d \mid \pi(x) > 0\}$. When $\pi$ satisfies:

\begin{itemize} 
    \item[$-$]$C1$: $\supp(\pi)$ is the unit box (or unit ball, resp.); 
    \item[$-$]$C2$: for every point $x \in \mathit{supp}(\pi)$, it holds that $\pi(x) = \Omega(1)$; 
\end{itemize}

\noindent
we will show that $(X, \R|_X)$ has a $(\rho, \eps)$-summary whose size is asymptotically the same as the aforementioned bound for box uniform (or ball uniform, resp.). Conditions $C1$ and $C2$ are satisfied by many distributions encountered in practice (e.g., the truncated versions of the Gaussian, Elliptical, and Laplace distributions, etc.), suggesting that real-world datasets may have much smaller $(\rho, \eps)$-summaries  than previously thought. 

\begin{figure}
  \centering
      \includegraphics[height=35mm]{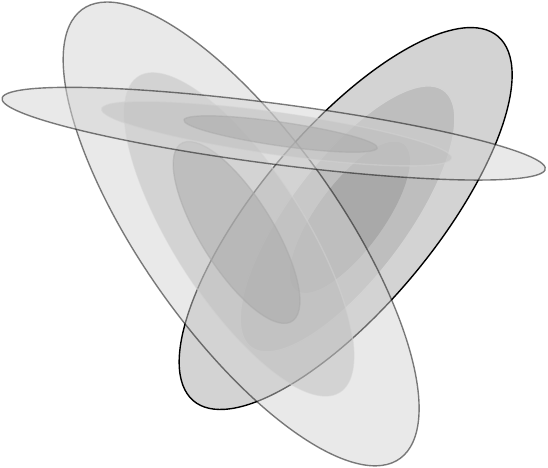}
  \figcapup 
  \caption{Mixture of 3 truncated Gaussian distributions in 2D space} 
  \label{fig::intro-mix}
  \figcapdown 
\end{figure}  

\vgap 

Even better, the linearity of halfspaces implies that, the same bounds still hold even when the shape of $\supp(\pi)$ in Condition $C1$ is obtained from the unit box/ball by an affine transformation.
Call a distribution {\em atomic} if it satisfies $C1$ (perhaps after an affine transformation) and $C2$. Our 
results hold also on ``composite distributions'' synthesized from a constant $z$ number of atomic distributions whose support regions may overlap {\em arbitrarily}. Specifically, let $\pi_1, \pi_2, ..., \pi_z$ be the pdfs of atomic distributions; and define $\pi(x) = \sum_{i=1}^z \gamma_i \cdot \pi(x)$, for arbitrary positive constants $\gamma_1, \gamma_2, ..., \gamma_z$ that sum up to 1; see Figure~\ref{fig::intro-mix} for an example. Then, the $(\rho,\eps)$-summary bound on $\pi$ is asymptotically determined by the highest of the $(\rho,\eps)$-summary bounds on $\pi_1, ..., \pi_c$. 

\vgap

\section{Disagreement Coefficients} \label{sec::coeff}

\noindent {\bf Existing Definitions on Distributions.} Disagreement coefficient was introduced by Hanneke \cite{h07} to analyze active learning algorithms (although a similar concept had been coined earlier \cite{a87} in statistics). 

\vgap

Let $\D$ be a distribution over $\real^d$. For any region $A \subseteq \real^d$, we denote by $\Pr_\D[A]$ the probability of $x \in A$ when $x$ is drawn from $\D$. Let $\R$ be a family of geometric ranges. Given a subset $\R' \subseteq \R$, define the {\em disagreement region} $\dis(\R')$ of $\R'$ as 
\begin{eqnarray} 
    \dis(\R') &=& \{x \in \real^d \mid \exists h_1, h_2 \in \R' \textrm{ s.t.\ } x \in h_1 \textrm{ and } x \notin h_2\}. \nn
\end{eqnarray}
That is, $\dis(\R')$ includes every such point $x \in \real^d$ that does not fall in all the ranges in $\R'$, and in the meantime, does not fall outside all the ranges in $\R'$, either. Given a range $h \in \R$ and a real value $r > 0$, define its {\em $r$-ball} $B_\D(h, r)$ as the set of all ranges $h' \in \R$ satisfying $\Pr_\D[\dis(\{h, h'\})] \leq r$. It is worth mentioning that $\dis(\{h, h'\})$ is simply the symmetric difference between $h$ and $h'$.

\vgap

Now, fix a range $h$, and consider increasing $r$ continuously; this can only expand the set $B_\D(h, r)$, and hence, also $\dis(B_\D(h, r))$. Interestingly, even though $\Pr_\D[\dis(B_\D(h, r))]$ is monotonically increasing, the ratio $\Pr_\D[\dis(B_\D(h, r))] / r$ may remain bounded by a certain quantity. Given a real value $\sigma \ge 0$, the {\em disagreement coefficient} $\theta_\D^h(\sigma)$ of $h$ measures this quantity with respect to all $r > \sigma$:
\begin{eqnarray} 
    \theta_\D^h(\sigma) &=& \max\left\{1, \sup_{r>\sigma} \fr{\Pr_\D[\dis(B_\D(h, r))]}{r}\right\}. \label{eqn::coeff-dist_coeff}
\end{eqnarray}
The function $\theta_\D^h(\sigma)$ has several useful properties: 

\begin{enumerate} 
    \item By definition, $\theta_\D^h(\sigma)$ is between 1 and $1/\sigma$, regardless of $\D$ and $h$. 
    
    \vgap 
    
    \item The supremum in \eqref{eqn::coeff-dist_coeff} ensures that $\theta_\D^h(\sigma)$ is monotonically decreasing. 
    
    \vgap
    
    \item For any $c \ge 1$, it holds that $\theta^h_\D(\sigma) \le c \cdot \theta^h_\D(c \sigma)$  (see Corollary 7.2 of \cite{h14b}).
\end{enumerate}

\extraspacing {\bf New Definitions on Range Spaces.} The above definitions rely on $\D$, and are not suitable for our problem settings where the input $X$ is a finite set. Next, we present a way to adapt the definitions to a range space $(X, \R|_X)$ for analyzing geometric algorithms. 

\vgap

We impose a uniform distribution over $X$: let $\U(X)$ be the distribution of a random point drawn uniformly from $X$. By replacing $\D$ with $\U(X)$ in \eqref{eqn::coeff-dist_coeff}, we rewrite \eqref{eqn::coeff-dist_coeff} into the following for any $\sigma \ge 0$: 
\begin{eqnarray} 
    \theta_{\U(X)}^h(\sigma) &=& \max\left\{1, \sup_{r>\sigma} \fr{\Pr_{\U(X)}[\dis(B_{\U(X)}(h, r))]}{r}\right\}. \label{eqn::coeff-set_coeff_range}
\end{eqnarray}
Set 
\begin{eqnarray} 
    \sigma_{min} &=& \fr{\min_{h \in \R} |X \cap h|}{n} \label{eqn::coeff-sigma_min}
\end{eqnarray}
We define the {\em disagreement coefficient of the range space $(X, \R|_X)$} as a function $\theta_{X}(\sigma): [\sigma_{min}, \infty) \rightarrow \real$ where
\begin{eqnarray} 
    \theta_{X}(\sigma) &=& \min_{h \in \R \textrm{ s.t. } \ol{X}(h) \le \sigma} \left\{\theta_{\U(X)}^h(\sigma)\right\}. \label{eqn::coeff-set_coeff}
\end{eqnarray}
It is clear from the above discussion that $1 \le \theta_{X}(\sigma) \le 1/\sigma$ and $\theta_{X}(\sigma)$ is monotonically decreasing.


\vgap

As a remark, the finiteness of $X$ gives a simpler interpretation of the $r$-ball $B_{\U(X)}(h, r)$: it is the set of ranges $h' \in \R$ such that $\dis(\{h,h'\})$ covers no more than $r |X|$ points in $X$. Also, $\Pr_{\U(X)}[A]$ for any region $A \subseteq \real^d$ is simply $|X \cap A| / |X|$.


\section{Small \bm{$(\rho, \eps)$}-Summaries Based on Disagreement Coefficients} \label{sec::summary} 

Given a range space $(X, \R)$ with VC-dimension $\lambda$, we will show how to find a $(\rho, \eps)$-summary whose size can be bounded using disagreement coefficients. Our algorithm is randomized, and succeeds with probability at least $1-\delta$ for a real-valued parameter $0 < \delta < 1$. Set $n = |X|$. We require that $\rho \ge \sigma_{min}$; otherwise, manually increasing $\rho$ to $\sigma_{min}$ achieves the same approximation guarantee because (i) no ranges $h$ can satisfy $0 < \ol{X}(h) < \sigma_{min}$, and (ii) as shown below, our summary definitely returns an estimate 0 if $\ol{X}(h) = 0$. 


\subsection{Algorithms} \label{sec::summary-algo}

\subsubsection{Computing a \bm{$(\rho,\eps)$}-Summary} \label{sec::summary-algo-1_build}

We will shrink $\R$ progressively by removing a range $h$ from $\R$ once we are sure we can provide an accurate estimate for $\ol{X}(h)$. Define $\R_0 = \R$. We perform at most $\lc \log (1/\rho) \rc$ rounds. Given $\R_{i-1}$, Round $i \ge 1$ is executed as follows:
\begin{center}
    \begin{minipage}{.8\linewidth}
        \begin{tabbing} 
            \tabpos 
            \> 1. $m_i \leftarrow$ the number of points $x \in X$ such that $x$ falls in {\em all} the ranges in $\R_{i-1}$ \\ 
            \> 2. $X_{i} \leftarrow X \cap \dis(\R_{i-1})$ \\ 
            \> 3. draw a set $S_i$ of points uniformly at random from $X_i$ with \\
            \begin{minipage}{\linewidth}
            \begin{eqnarray} 
                \hspace{5mm} |S_i| &=& O\left(\fr{|X_i|}{n} \cdot \fr{2^i}{\eps^2} \left(\lambda \log \fr{1}{\rho} + \log \fr{\log(1/\rho)}{\delta} \right) \right) \label{eqn::summary-each_step_sample_size}
            \end{eqnarray}
            \end{minipage}
            \\[4mm]
            \> 4. $\R_i = \{h \in \R_{i-1} \,\,\mid\,\, \ol{S_i}(h) \cdot |X_i| + m_i < n / 2^i\}$ 
        \end{tabbing}
    \end{minipage}
\end{center}
The algorithm terminates when either $i = \lc \log (1/\rho) \rc$ or $\R_i = \emptyset$. Suppose that in total $t$ rounds are performed. The final $(\rho,\eps)$-summary consists of sets $S_1, S_2, ..., S_t$, and $2t + 1$ integers $n, m_1, m_2, ..., m_t, |X_1|, |X_2|, ..., |X_t|$. 

\subsubsection{Performing Estimation} \label{sec::summary-algo-2_est}

Given a range $h \in \R$, we deploy the summary to estimate $\ol{X}(h)$ in two steps:
\begin{center}
    \begin{minipage}{.8\linewidth}
        \begin{tabbing} 
            \tabpos 
            \> 1. $j \leftarrow$ the largest $i \in [1, t]$ such that $h \in \R_i$ \\ 
            \> 2. return $\ol{S_j}(h) \cdot \fr{|X_j|}{n} + \fr{m_j}{n}$ as the estimate
        \end{tabbing}
    \end{minipage}
\end{center}
Regarding Step 1, whether $h \in \R_i$ can be detected as follows. First, if $h \notin \R_{i'}$ for any $i' < i$, then immediately $h \notin \R_i$. Otherwise, compute $\ol{S_i}(h)$, and declare $h \in \R_i$ if and only if $\ol{S_i}(h) \cdot |X_i| + m_i < n/2^i$. 

\subsection{Analysis} \label{sec::summary-analysis}

We now proceed to prove the correctness of our algorithms, and bound the size of the produced summary. It suffices to consider $\eps \le 1/3$ (otherwise, lower $\eps$ to $1/3$ and then apply the argument below). 

\vgap

The subsequent discussion is carried out under the event that, for every $i \in [1, t]$, $S_i$ is a relative $(\rho_i,\eps/4)$-approximation of $X_i$ with respect to the ranges in $\R$ where 
\begin{eqnarray} 
    \rho_i &=& \fr{n(1+\eps)}{2^i \cdot |X_i|}. \nn
\end{eqnarray}
By the result of \cite{lls01} (reviewed in Section~\ref{sec::intro-prev}), with $|S_i|$ shown in \eqref{eqn::summary-each_step_sample_size}, the event happens with a probability at least $1 - \delta \cdot \fr{t}{\lc \log(1/\rho) \rc} \ge 1 - \delta$. 

\subsubsection{Correctness} 

To show that our algorithm indeed outputs a $(\rho,\eps)$-summary, the key step is to prove: 

\begin{lemma} \label{lmm::summary-rel_approx}
    The following are true for all $i \in [1, t]$: (i) for every range $h \in \R_i$, $\ol{X}(h) < (1+\eps)/2^i$; (ii) for every range $h \notin \R_i$, $\ol{X}(h) \ge (1-\eps)/2^i$. 
\end{lemma}

\begin{proof} 
    See Appendix~\ref{app::lmm_rel_approx}.
\end{proof}

Now consider the estimation algorithm in Section~\ref{sec::summary-algo-2_est}. Given the value $j$ obtained at Step 1 for the input range $h \in \R$, the above lemma suggests that \[(1-\eps)/2^{j+1} \le \ol{X}(h) < (1+\eps)/2^j.\] This, together with $S_j$ being a $(\rho_j,\eps/4)$-approximation of $X_j$, ensures that our estimate satisfies the $(\rho, \eps)$-guarantee for $h$. The details can be found in Appendix~\ref{app::qry_correct}.

\subsubsection{Bounding the Size} 

To bound the size of our $(\rho,\eps)$-summary, we will focus on bounding $\sum_{i=1}^t |S_i|$, because the rest of the summary clearly needs  $O(t) = O(\log(1/\rho))$ extra integers. Let us start with a trivial bound that follows directly from $|X_i|\le n$: 
\begin{eqnarray}
    \sum_{i=1}^t |S_i| 
    &=&
    O\left(\sum_{i=1}^t \fr{2^i}{\eps^2} \left(\lambda \log \fr{1}{\rho} + \log \fr{\log(1/\rho)}{\delta} \right) \right) \nn \\
    &=& O\left(\fr{1}{\rho \eps^2} \left(\lambda \log \fr{1}{\rho} + \log \fr{\log(1/\rho)}{\delta} \right) \right). \label{eqn::summary-trivial}
\end{eqnarray}

Next, we use disagreement coefficients to prove a tighter bound. Fix $h \in \R$ to be an {\em arbitrary} range such that $\ol{X}(h) \le \rho$ ($h$ definitely exists because $\rho \ge \sigma_{min}$). 


\begin{lemma} \label{lmm::summary-size_1_help}
    $\R_i \subseteq B(h, \rho+(1+\eps)/2^i)$.
\end{lemma}

\begin{proof}
    It suffices to prove that, for any $h' \in \R_i$, $\Pr_{\U(X)}[\dis(\{h, h'\})] \le \rho + (1+\eps)/2^i$, or equivalently, $|X \cap \dis(\{h, h'\})| \le n (\rho + (1+\eps)/2^i)$. 
    
    \vgap
    
    This holds because $|X \cap \dis(\{h, h'\})| \le |(X \cap h) \cup (X \cap h')|$. By definition of $h$, we know $|X \cap h| \le n \rho$, while By Lemma~\ref{lmm::summary-rel_approx}, we know $|X \cap h'| \le n (1+\eps)/2^i$. Therefore, $|X \cap \dis(\{h, h'\})| \le n (\rho + (1+\eps)/2^i)$. 
\end{proof}

\begin{lemma} \label{lmm::summary-size_2_dis_standard}
    $|X_i| / n \le 
    \theta^h_{\U(X)}(2\rho) \cdot (\rho+\fr{1+\eps}{2^{i-1}}).$ 
\end{lemma}

\begin{proof}
    Lemma~\ref{lmm::summary-size_1_help} tells us that $\dis(\R_{i-1}) \subseteq \dis(B_{\U(X)}(h, \rho+\fr{1+\eps}{2^{i-1}}))$. Thus:
\begin{eqnarray} 
    |X_i| / n &=& \Pr_{\U(X)} (\dis(\R_{i-1})) \nn \\ 
    &\le& \Pr_{\U(X)} 
    \left(\dis
    \left(B_{\U(X)}
    \left(h, \rho+ \fr{1+\eps}{2^{i-1}}
    \right)
    \right)
    \right)\nn \\
    \textrm{(by \eqref{eqn::coeff-set_coeff_range})}
    &\le&
    \theta^h_{\U(X)}
    \left(\rho + \fr{1+\eps}{2^{i-1}}
    \right) \cdot 
    \left(\rho + \fr{1+\eps}{2^{i-1}}
    \right) \nn 
\end{eqnarray}    
By $1/2^{i-1} > \rho$, and the fact that $\theta^h_{\U(X)}$ is monotonically decreasing, the above leads to
\begin{eqnarray}
    \theta^h_{\U(X)}
    \left(\rho + \fr{1+\eps}{2^{i-1}}
    \right) \cdot 
    \left(\rho + \fr{1+\eps}{2^{i-1}}
    \right)
    &\le&
    \theta^h_{\U(X)}(\rho + \rho) \cdot \left(\rho + \fr{1+\eps}{2^{i-1}}
    \right) \nn \\
    &=&
    \theta^h_{\U(X)}(2\rho) \cdot \left(\rho + \fr{1+\eps}{2^{i-1}}
    \right). \nn
\end{eqnarray}
\end{proof}

Therefore: 
\begin{eqnarray}
    \sum_{i=1}^t \fr{|X_i| \cdot 2^i}{n} 
    &\le&
    \theta^h_{\U(X)}(2\rho) \cdot \sum_{i=1}^t 2^i \cdot \left(\rho+ \fr{1+\eps}{2^{i-1}} \right) \nn \\
    &=&
    \theta^h_{\U(X)}(2\rho) \cdot \sum_{i=1}^t \left(2^i \cdot \rho+ O(1) \right) \nn \\
    \textrm{(by $1/2^i = \Omega(\rho)$)}
    &=& 
    \theta^h_{\U(X)}(2\rho) \cdot O(t) \nn \\
    &=& \theta^h_{\U(X)}(2\rho) \cdot O(\log(1/\rho)) \nn \\
    &=& 
    \theta^h_{\U(X)}(\rho) \cdot O(\log(1/\rho)) 
    \label{eqn::summary-bound_one_range}
\end{eqnarray}
where the last equality used the fact that $\theta^h_{\U(X)}(2\rho) \le 2 \cdot \theta^h_{\U(X)}(\rho)$.

\vgap

Remember that the above holds for {\em all} $h \in \R$ satisfying $\ol{X}(h) \le \rho$. By the definition in \eqref{eqn::coeff-set_coeff}, we can improve the bound of \eqref{eqn::summary-bound_one_range} to
\begin{eqnarray}
    \sum_{i=1}^t \fr{|X_i| \cdot 2^i}{n}
    =
    \theta_{X}(\rho) \cdot O(\log(1/\rho)). 
\end{eqnarray}
Combining the above with \eqref{eqn::summary-each_step_sample_size} gives $\sum_{i=1}^t |S_i| = O(\fr{1}{\eps^2} \cdot \theta_{X}(\rho) \log(1/\rho) \cdot (\lambda \log(1/\rho) + \log \fr{\log(1/\rho)}{\delta}))$. Putting this together with \eqref{eqn::summary-trivial} and setting $\delta$ to a constant gives: 

\begin{theorem} \label{thm::summary-main}
    For any $\rho \ge \sigma_{min}$ and any $0 < \eps < 1$, a range space $(X, \R|_X)$ of VC-dimension $\lambda$ has a $(\rho,\eps)$-summary which keeps $O(\log(1/\rho))$ integers and 
    $O(\min\big\{\fr{1}{\rho}, \theta_{X}(\rho) \cdot \log \fr{1}{\rho} \big\} \cdot \fr{\lambda }{\eps^2}  \log \fr{1}{\rho} )$ 
    points of $X$. Here, $\sigma_{min}$ is defined in \eqref{eqn::coeff-sigma_min}, and $\theta_{X}$ is the disagreement coefficient function defined in \eqref{eqn::coeff-set_coeff}.
\end{theorem}

\subsubsection{A Remark} 


Our $(\rho,\eps)$-summary is currently not sample-based, but this can be fixed by keeping --- at Step 1 of the computation algorithm in Section~\ref{sec::summary-algo} --- an arbitrary point counted by $m_i$. 

\vgap

The $(\rho,\eps)$-summary after the fix also serves as a $\rho$-net. Thus, by setting $\eps$ to a constant in Theorem~\ref{thm::summary-main}, we know that for any $\rho \ge \sigma_{min}$, the range space $(X, \R|_X)$ in Theorem~\ref{thm::summary-main} has an $\rho$-net of size $O(\min\big\{\fr{1}{\rho}, \theta_{X}(\rho) \cdot \log \fr{1}{\rho} \big\} \cdot \lambda \log \fr{1}{\rho})$. However, it should be pointed out that this bound on $\rho$-nets can be slightly improved, as is implied by Theorem 5.1 of \cite{h14b} and made explicit in \cite{kz17}. 

\section{Bridging Distribution and Finite-Set Disagreement Coefficients} \label{sec::bridge}

This section will establish another theorem which will be used together with Theorem~\ref{thm::summary-main} to explain why we are able to obtain $(\rho,\eps)$-summarizes of $o(1/\rho)$ size on practical datasets. Suppose that the input $X$ has been generated by taking $n$ points independently following the same distribution $\D$ over $\real^d$. The learning literature (see, e.g., \cite{h14b}) has developed a solid understanding on when the quantity $\theta^h_\D(\sigma)$ is small. Unfortunately, those findings can rarely be applied to $\theta^h_{\U(X)}(\sigma)$ because they are conditioned on requirements that must be met by $\D$, e.g., one common requirement is continuity. $\U(X)$, due to its discrete nature, seldom meets the requirements. 

\vgap

On the other hand, clearly $\U(X)$ approximates $\D$ increasingly better as $n$ grows. Thus, we ask the question:

\begin{center}
    How large $n$ needs to be before $\theta^h_{\U(X)}(\sigma)$ is asymptotically the same as $\theta^h_\D(\sigma)$?
\end{center}

\noindent We give an answer in the next theorem:

\begin{theorem} [The Bridging Theorem] \label{thm::bridge-bridge}
     Let $\D$ be a distribution over $\real^d$, and $\R$ be a family of ranges. Denote by $\lambda$ the VC-dimension of the range space $(\real^d, \R)$. 
     
     Fix an arbitrary range $h \in \R$, an arbitrary integer $n$, a real value $0 < \delta < 1$, a real value $\sigma$ satisfying $n \ge \fr{c}{\sigma} (\log \fr{n}{\delta} + \lambda \log \fr{1}{\sigma})$ for some universal constant $c$. If we draw a set $X$ of $n$ points independently from $\D$, then with probability at least $1-\delta$, it holds that $\theta_{\U(X)}^h(\sigma) \le 8 \cdot \theta^h_\D (2 \sigma)$. 
\end{theorem}

The rest of the section serves as a proof of the theorem. Let us first get rid of two easy cases: 
\begin{itemize} 
    \item[$-$] If $\sigma \ge 1$, $\theta_{\U(X)}^h(\sigma) = \theta^h_\D (2 \sigma) = 1$ by definition of \eqref{eqn::coeff-set_coeff_range}; and the theorem obviously holds.
    
    \item[$-$] If $\sigma < 1/n$, observe that every range $h' \in B_{\U(X)}(h, \sigma)$ covers exactly the same set of points in $X$ as $h$. Hence, $\Pr_{\U(X)}[\dis(B_{\U(X)}(h, r))] = 0$. It follows from \eqref{eqn::coeff-set_coeff_range} that $\theta_{\U(X)}^h(\sigma) = 1$. The theorem again obviously holds because $\theta^h_\D (2 \sigma) \ge 1$, by definition.
\end{itemize}

Hence, it suffices to consider $1/n \le \sigma < 1$. Define $S = \{i/n \mid$ $i$ is an integer in $[\sigma n, n]\}$. For $\sigma \ge 1/n$, \eqref{eqn::coeff-set_coeff_range} implies
\begin{eqnarray} 
    \theta_{\U(X)}^h(\sigma) &\le& \max\left\{1, 2 \cdot \max_{r \in S} \fr{\Pr_{\U(X)}[\dis(B_{\U(X)}(h, r))]}{r}\right\}. \label{eqn::bridge-set_coeff_range_double}
\end{eqnarray}
Consider an arbitrary $r \in S$. We will show that, when $n$ satisfies the condition in the theorem, with probability at least $1 - \delta/n$, it holds that 
\begin{eqnarray}
    \fr{\Pr_{\U(X)}[\dis(B_{\U(X)}(h, r))]}{r} 
    &\le& 4 \cdot \theta^h_\D(2\sigma). \label{eqn::bridge-target-0}
\end{eqnarray}
Once this is done, applying the union bound on all the $r \in S$ will prove that \eqref{eqn::bridge-set_coeff_range_double} is at most $8 \cdot \theta^h_\D (2 \sigma)$ with probability at least $1-\delta$, as claimed in the theorem.

\vgap

We aim to establish the following equivalent form of \eqref{eqn::bridge-target-0}:
\begin{eqnarray}
    \fr{|X \cap \dis(B_{\U(X)}(h, r))|}{nr} \le 4 \cdot \theta^h_\D(2\sigma). \label{eqn::bridge-target-1}
\end{eqnarray}

For the above purpose, the most crucial step is to prove:
\begin{lemma} \label{lmm::bridge-step-1}
    When $n \ge \fr{c_1}{r} (\lambda \log \fr{1}{r} + \log \fr{n}{\delta})$ for some universal constant $c_1$, it holds with probability at least $1-\delta/(2n)$ that
    $B_{\U(X)}(h, r) \subseteq B_\D(h, 2r)$.
\end{lemma}

\begin{proof}
    The rationale of our proof is that any $h' \notin B_\D(h, 2r)$ is unlikely to appear in $B_{\U(X)}(h, r)$ when $n$ is large. Indeed, $h' \notin B_\D(h, 2r)$ indicates that a point $x$ drawn from $D$ has probability over $2r$ to fall in $\dis(\{h,h'\})$. Hence, $|X \cap \dis(\{h,h'\})|$ should be sharply concentrated around $2r\cdot n$, rendering $h' \notin B_{\U(X)}(h, r)$. The challenge, however, is that there can be an infinite number of ranges $h'$ to consider. To tackle the challenge, we need to bring down the number of ranges somehow to $n^{O(\lambda)}$. We achieve the purpose by observing that we can define another range space with VC-dimension $O(\lambda)$ to capture the disagreement regions of range pairs from $\R$, as shown below.
    
    \vgap

    Define $\R^\mathit{dis} = \{\dis(\{h, h'\}) \mid h, h' \in \R\}$. We observe that the range space $(\real^d, \R^\mathit{dis})$ has VC-dimension $O(\lambda)$. To explain why, for any $h \in \R$, define $\ol{h} = \real^d \setminus h$. Accordingly, define $\ol{\R} = \{\ol{h} \mid h \in \R\}$. The two range spaces $(\real^d, \R)$ and $(\real^d, \ol{\R})$ have the same VC-dimension $\lambda$. Therefore, the range space $(\real^d, \R \cup \ol{\R})$ has VC-dimension at most $2\lambda + 1$. Now apply a 2-fold intersection on $(\real^d, \R \cup \ol{\R})$ to create $(\real^d, \R_1)$ where $\R_1 = \{h \cap h' \mid h, h' \in \R \cup \ol{\R}\}$. By a result of \cite{behw89}, the VC dimension of $(\real^d, \R_1)$ is bounded by $O(\lambda)$. Finally, apply a 2-fold union on $(\real^d, \R_1)$ to create $(\real^d, \R_2)$ where $\R_2 = \{h \cup h' \mid h, h' \in \R_1\}$. By another result of \cite{behw89}, the VC dimension of $(\real^d, \R_2)$ is bounded by $O(\lambda)$. Notice that $\R^\mathit{dis}$ is a subset of $\R_2$. It thus follows that the VC-dimension of $(\real^d, \R^\mathit{dis})$ must be $O(\lambda)$.
    
    \vgap 
    
    Essentially, now the task is to draw a sufficiently large set $X$ of points from $\D$ to guarantee with probability at least $1 - \delta/(2n)$: for every range $h \in \R^\mathit{dis}$ with $\Pr_\D(h) > 2r$, we ensure $|X \cap h|/|X| > r$. By applying a result of \cite{lls01} on general range spaces, we know that $|X|$ only needs to be $\fr{c_1}{r} (\lambda \log \fr{1}{r} + \log\fr{n}{\delta})$ for some constant $c_1$ which does not depend on $r, \delta$, and $n$. 
\end{proof}

Set $r' = \Pr_\D(\dis(B_{\D}(h, 2r))$; notice that, by definition of $\theta^h_\D(2 r)$, $r' \le 2r \cdot \theta^h_\D(2 r)$. We want to draw a sufficiently large set $X$ of points from $\D$ to guarantee, with probability at least $1 - \delta/(2n)$, $|X \cap \dis(B_{\D}(h, 2r))| \le 2n \cdot \max\{r, r'\}$. By Chernoff bounds, $n$ only needs to be at least $\fr{c_2}{r} \log \fr{n}{\delta}$ for some universal constant $c_2$.

\vgap 

Now, set $c = \max\{c_1, c_2\}$ and $n = \fr{c}{r} (\lambda \log \fr{1}{r} + \log\fr{n}{\delta})$. With probability at least $1-\delta/n$, we can derive \eqref{eqn::bridge-target-1} from the above discussion as follows: 
\begin{eqnarray} 
    \fr{|X \cap \dis(B_{\U(X)}(h, r))|}{nr}
    &\le& 
    \fr{|X \cap \dis(B_{\D}(h, 2r))|}{nr} \textrm{\hspace{5mm} (by Lemma~\ref{lmm::bridge-step-1})} \nn \\ 
    &\le& 
    \fr{2 n\cdot \max\{r, r'\} }{nr} 
    = 2 \cdot \max\{1, r'/r\} \nn \\
    &\le& 2 \cdot \max\{1, 2 \cdot \theta^h_\D(2 r)\} 
    = 4 \cdot \theta^h_\D(2 r) 
    \le 4 \cdot \theta^h_\D(2 \sigma) \nn 
\end{eqnarray}
where the last inequality used $r \ge \sigma$ and the fact that $\theta^h_\D$ is monotonically decreasing. This establishes \eqref{eqn::bridge-target-1} and hence completes the proof of Theorem~\ref{thm::bridge-bridge}.

\section{\bm{$o(1/\rho)$}-Size Summaries for Halfspace Ranges} \label{sec::halfspace} 

We are ready to explain why a set of points generated from a stochastic distribution often admits $(\rho,\eps)$-summaries of $o(1/\rho)$ size for fixed $\eps$. This requires specializing $\R$ into a concrete range family. We will do so by constraining $\R$ to be the set of halfspaces in $\real^d$, because this family has received considerable attention (as reviewed in Section~\ref{sec::intro-prev}). 

\vgap 

We prove in the appendix the next two technical lemmas regarding the disagreement coefficients on box-uniform and ball-uniform distributions:
\begin{lemma} \label{lmm::halfspace-dist_coeff-box}
    Let $\U$ be the distribution where a point is drawn uniformly at random from the unit box $[0, 1]^d$ with $d = O(1)$. For any halfspace $h$ disjoint with the box, it holds that $\theta^h_\U(\sigma) = O(\log^{d-1} \fr{1}{\sigma})$ for all $\sigma > 0$. 
\end{lemma}
\begin{lemma} \label{lmm::halfspace-dist_coeff-ball}
    Let $\U$ be the distribution where a point is drawn uniformly at random from the unit ball $\{x \in \real^d \mid \sum_{i=1}^d x[i]^2 \le 1\}$ with $d = O(1)$. For any halfspace $h$ disjoint with the ball, it holds that $\theta^h_\U(\sigma) = O((\fr{1}{\sigma})^\fr{d-1}{d+1})$ for all $\sigma > 0$. 
\end{lemma}

Next, we establish our main result for {\em non-uniform} distributions: 

\begin{theorem} \label{thm::halfspace-summary}
    Let $\R$ be the family of halfspaces in $\real^d$ with a constant dimensionality $d$. Let $\D$ be a distribution over $\real^d$ such that the pdf $\pi$ of $\D$ satisfies Conditions $C1$ and $C2$ as prescribed in Section~\ref{sec::intro-ours}.    
    Suppose that we draw a set $X$ of $n$ points independently from $\D$. Both of the following hold with probability at least $1 - 1/n^2$:

    \begin{itemize} 
        \item[$-$] When $\supp(\pi)$ is the unit box, for any $0 < \eps < 1$ and any $\rho \ge \fr{c \log n}{n}$ where $c > 0$ is a constant, $X$ has a $(\rho,\eps)$-summary that keeps $O(\log(1/\rho))$ integers and $O(\fr{1}{\eps^2} \log^{d+1} \fr{1}{\rho})$ points of $X$. 
     
        \item[$-$] When $\supp(\pi)$ is the unit ball, for any $0 < \eps < 1$ and any $\rho \ge \fr{c \log n}{n}$ where $c > 0$ is a constant, $X$ has a $(\rho,\eps)$-summary that keeps $O(\log(1/\rho))$ integers and $O(\fr{1}{\eps^2} \cdot (\fr{1}{\rho})^\fr{d-1}{d+2} \cdot \log^2 \fr{1}{\rho})$ points of $X$. 
    \end{itemize}
    
    \noindent The constant $c$ in the above does not depend on $\D$, $n$, $\rho$, and $\eps$.
\end{theorem}

\begin{proof} 
    We will prove only the case where $\supp(\pi)$ is the unit box because the unit-ball case is similar. 
    Set $\sigma^* = \fr{c \cdot \log n}{n}$ where $c$ is some constant to be determined later. Thanks to Theorem~\ref{thm::summary-main}, it suffices to prove that with probability at least $1 - 1/n^2$, $\theta_{X}(\rho) = O(\log^{d-1} \fr{1}{\rho})$ at every $\rho \ge \sigma^*$. We will argue that, with probability at least $1-1/n^2$, there exists a halfspace $h \in \R$ such that $\ol{X}(h) \le \rho$ and $\theta^h_{\U(X)} (\rho) = O(\log^{d-1} \fr{1}{\rho})$. Once this is done, we know $\theta_{X}(\rho) = O(\log^{d-1} \fr{1}{\rho})$ from \eqref{eqn::coeff-set_coeff}.
    
    \vgap 
    
    Condition $C_2$ says that the pdf $\pi$ satisfies $\pi(x) \ge \gamma$ for any point $x$ in $\supp(\pi)$ (i.e., the unit box), where $\gamma$ is a positive constant. Remember that, by definition of $\supp(\pi)$, $\pi(x) = 0$ for any $x$ outside $\supp(\pi)$. 
    
    \vgap 
    
    Simply set $h$ to a halfspace as stated in Lemma~\ref{lmm::halfspace-dist_coeff-box}, i.e., $\theta^h_\U(\sigma) = O(\log^{d-1} \fr{1}{\sigma})$. Let $\pi_\U$ be the pdf of $\U$: $\pi_\U(x)$ equals 1 if $x \in [0, 1]^d$, or 0 otherwise. Define $\alpha$ as any constant such that $\alpha \le \gamma$. We have $\alpha \cdot \pi_\U(x) \le \pi(x) \le 1 \le \fr{1}{\alpha} \cdot \pi_\U(x)$ for all $x \in \real^d$. Given this, Theorem 7.6 of \cite{h14b} tells us that $\theta^h_\D(\sigma) = O(\theta^h_\U(\sigma / \alpha))$. It thus follows that $\theta^h_\D(\sigma) = O(\log^{d-1} \fr{1}{\sigma})$ for all $\sigma > 0$. 
    
    \vgap 
    
    Now, apply Theorem~\ref{thm::bridge-bridge} on $h$ by setting $\delta = 1/n^2$ and $\lambda = O(1)$. The theorem shows that, when $n \ge \fr{\beta \cdot \log n}{\rho}$ for some constant $\beta$, $\theta^h_{\U(X)}(\rho) \le 8 \cdot \theta^h_D(\rho) = O(\log^{d-1} \fr{1}{\rho})$ with probability at least $1 - 1/n^2$. We set $c \ge \beta$ to ensure $n \ge \fr{\alpha \cdot \log n}{\rho}$. Note also that the choice of $h$ guarantees $\ol{X}(h) = 0 < \rho$. This makes $h$ a halfspace we are looking for, and concludes the proof. 
\end{proof}

\noindent {\bf Remark 1 (Composite Distributions).} Let $\D_1$ and $\D_2$ be two distributions over $\real^d$ with pdfs $\pi_1$ and $\pi_2$, respectively (the support regions of $\pi_1$ and $\pi_2$ may overlap). Define a distribution $\D$ with pdf $\pi(x) = \gamma \cdot \pi_1(x) + (1-\gamma) \cdot \pi_2(x)$, for some constant $0 < \gamma < 1$. Theorem 7.7 of \cite{h14b} tells us that, for any halfspace $h \in \R$ and any $\sigma > 0$, $\theta^h_\D(\sigma) \le \theta^h_{\D_1}(\fr{\sigma}{\gamma}) + \theta^h_{\D_2}(\fr{\sigma}{1-\gamma})$. 

\vgap

It thus follows from Lemma~\ref{lmm::halfspace-dist_coeff-box} that, when $\D_1$ and $\D_2$ are atomic distributions with support regions obtainable from the unit box through affine transformations, $\theta^h_\D(\sigma) = O(\log^{d-2} \fr{1}{\sigma})$ for any $h$ disjoint with $\supp(\D_1)\cup\supp(\D_2)$. The unit-box result of Theorem~\ref{thm::halfspace-summary} can be easily shown to hold on this $\D$ as well, by adapting the proof in a straightforward manner. The same is true for the unit-ball result of Theorem~\ref{thm::halfspace-summary}. All these results can now be  extended to a composite distribution synthesized from a constant number of atomic distributions (see Section~\ref{sec::intro-ours}).

\extraspacing {\bf Remark 2 (More Distribution with Near-Constant \bm{$\theta$}).} What is given in Lemma~\ref{lmm::halfspace-dist_coeff-box} is only one scenario where $\theta^h_\D(\sigma)$ is nearly a constant. There are other combinations of $\D$ and $\R$ where $\theta^h_\D(\sigma) = \tilde{O}(1)$ for all $h \in \R$; see \cite{f09,h07,h14b,yw12,w11} (in some of those combinations, $\R$ may not contain all the halfspaces in $\real^d$; e.g., a result of \cite{h07} concerns only the halfspaces whose boundary planes pass the origin). The proof of Theorem~\ref{thm::halfspace-summary} can be adapted to show that $X$ has a $(\rho,\eps)$-summary of size $\tilde{O}(1/\eps^2)$ with high probability when $\rho = \Omega(\max\{\fr{\log n}{n}, \min_{h \in \R} \Pr_\D(h)\})$. 

\extraspacing {\bf Remark 3 (Time Complexity).} In general, for any $X$, a $(\rho,\eps)$-summary (for the halfspace family $\R$ in constant-dimensional space) can be found in polynomial time even by implementing the algorithm of  Section~\ref{sec::summary-algo} naively. The time can be improved to $O(n \polylog n) + n^{1-\Omega(1)} \cdot s^{O(1)}$, where $s$ is the size of the returned summary, by utilizing specialized data structures \cite{c12,m92}. 

\section{A Lower Bound with Disagreement Coefficients} \label{sec::lb}

In this section, we will prove a lower bound on the sizes of $(\rho, 1/2)$-summaries in relation to disagreement coefficients. Our core result is:


\begin{theorem} \label{thm::lb}
    Let $\R$ be the family of all halfplanes in $\real^2$. Fix any integer $w$ as the number of bits in a word. Choose arbitrary integers $\eta$, $q$, and $k$ such that $\eta \ge 4$, $q$ is a multiple of $\eta$, and $1 \le k \le q / (4\eta)$.    
    There must exist a set $\mathcal{C}$ of range spaces $(X, \R|_X)$, each satisfying the following conditions: 
    \begin{itemize} 
        \item[$-$]$X$ is a set of $q + k$ points in $\real^2$.  
        \item[$-$]The disagreement coefficient of $(X, \R|_X)$ satisfies $\theta_{X}(\fr{k}{q+k}) = \fr{k + q}{k + q/\eta}$.
        \item[$-$]Any encoding, which encodes a $(\fr{k}{q+k}, 1/2)$-summary for each range space in $\mathcal{C}$, must use at least $\eta \cdot w$ bits on at least one range space in $\mathcal{C}$.
    \end{itemize}
\end{theorem}

Therefore, for $\rho = \fr{k}{q+k}$, if one wishes to store a $(\rho, 1/2)$-summary for each range space in $\mathcal{C}$, at least $\eta \cdot w$ bits (namely, $\eta$ words) are needed on at least one range space. Since $\theta_{X}(\rho) = \fr{k + q}{k + q/\eta} \le \eta$, this establishes $\theta_{X}(\rho)$ as a space lower bound for $(\rho, 1/2)$-summaries. In the theorem, any $X$ has dimensionality $d = 2$ and any $(X, \R)$ has VC-dimension at most 3; hence, Theorem~\ref{thm::summary-main} is tight up to polylog factors on constant $\lambda$ and $\eps$. 

\vgap

The flexibility of $\eta$, $q$, and $k$ allows the lower bound to hold in a variety of more concrete settings. For example, by adjusting $k$ and $q$, one sees that $\theta_{X}(\rho)$ is a lower bound for the whole range of $\rho \in (0, O(1)]$. On the other hand, by focusing on any specific $\rho \in (0, O(1)]$ but adjusting $\eta$, one sees that $\theta_{X}(\rho)$ remains as a lower bound when $\theta_{X}(\rho)$ goes from $O(1)$ to $\Omega(1/\rho)$.\footnote{This rules out, for example, a claim of the form: ``when $\theta_{X}(\rho) \ge \sqrt{1/\rho}$, there is a $(\rho, 1/2)$-summary of size $O(\sqrt{\theta_{X}(\rho)})$''.}

\extraspacing {\bf Proof of Theorem~\ref{thm::lb}.} Fix integers $\eta$, $q$, and $k$ as stated in Theorem~\ref{thm::lb}. Define $n = q + k$. Next, we construct a class $\X$ of point sets, each consisting of $n$ points in $\real^2$. First, place $k$ points at coordinates $(0, \infty)$. Call them the {\em outer} points; they belong to all the sets in $\mathcal{X}$. 

\vgap 

\begin{figure}
  \centering
      \includegraphics[height=45mm]{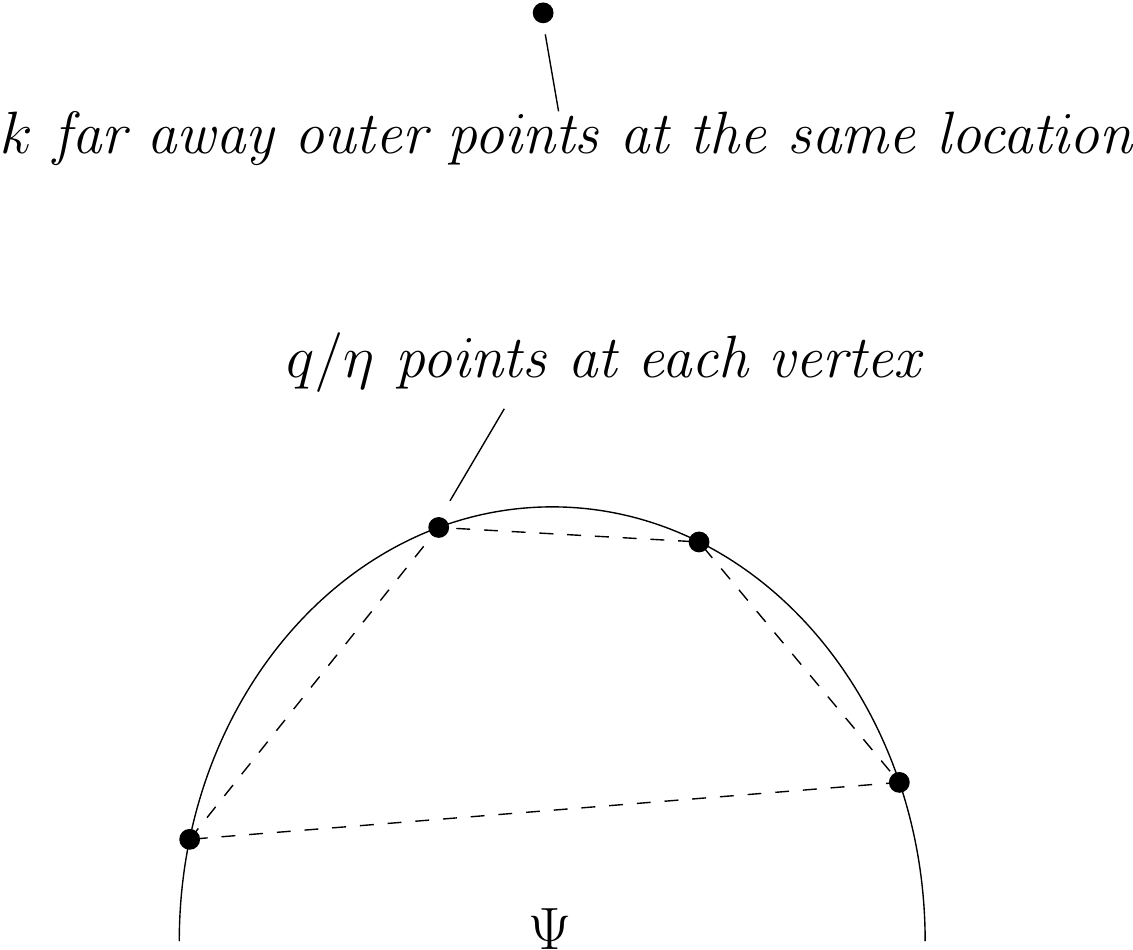}
  \figcapup 
  \caption{A set of points in $\X$ ($\eta = 4$)} 
  \label{fig::lb-lb}
  \figcapdown 
\end{figure}


For each set $X \in \X$, we generate $q$ extra {\em inner} points. For this purpose, place an arbitrary polygon $\Psi$ with $\eta$ vertices, making sure that all the vertices fall on the upper arc of the unit circle (i.e., the arc is $\{(x[1], x[2]) \mid x[1]^2 + x[2]^2 = 1$ and $x[2] \ge 0\}$). Then, given each vertex of $\Psi$, we add $q/\eta$ inner points to $X$, all of which are located at that vertex. See Figure~\ref{fig::lb-lb} for an example with $\eta = 4$. This finishes the construction of $X$. It is important to note that a different $\Psi$ is used for each $X$. Thus, $\X$ includes an infinite number of inputs, each corresponding to a possible $\Psi$. Our construction ensures a nice property: 

\begin{lemma}  \label{lmm::lb-infer-vertices}
    Fix any $X \in \mathcal{X}$. Given any $(k/n, 1/2)$-summary of $X$, we are able to infer all the vertices of $\Psi$ used to construct $X$. 
\end{lemma}
\begin{proof} 
    We say that a halfplane in $\R$ is {\em upward} if it covers the point $(0, \infty)$. Our aim is to prove that, the summary allows us to determine whether an arbitrary upward halfplane covers any inner point of $X$. This implies that we can reconstruct all the vertices of $\Psi$ using the summary.\footnote{To see this, consider any vertex $v$ of $\Psi$, and use the summary to distinguish the line $\ell$ tangent to the arc at $v$ and a line that is parallel to $\ell$, but moves slightly away from the arc.}
    
    \vgap
    
    Given an upward halfplane $h$, we use the summary to obtain an estimate --- denoted as $\tau$ --- of $\ol{X}(h)$. If $\tau \ge  2k/n$, we return ``yes'' (i.e., $h$ covers at least one inner point); otherwise, we return ``no''.     
    To see that this is correct, first note that $\ol{X}(h)$ must be at least $k/n$, and hence    
    $0.5 \ol{X}(h) \le \tau \le 1.5\ol{X}(h)$. Therefore, if $h$ covers no inner points, $\ol{X}(h) = k/n$, indicating $\tau < 1.5 k/n$. Otherwise, $\ol{X}(h) \ge \fr{k + q/\eta}{n} \ge 5k/n$, indicating $\tau \ge 2.5k/n$.
\end{proof}

We prove in the appendix: 

\begin{lemma} \label{lmm::lb-dis-coeff}
    For each $X \in \X$, the disagreement coefficient of $(X, \R|_X)$ satisfies $\theta_{X}(k/n) = \fr{n}{k+q/\eta}$.
\end{lemma}

The set $\mathcal{C}$ is simply the set $\{(X, \R) \mid X \in \X\}$. Recall that each $X \in \X$ corresponds to a distinct $\eta$-vertex polygon $\Psi$. Hence, by Lemma~\ref{lmm::lb-infer-vertices}, the $(k/n, 1/2)$-summaries associated with the range spaces in $\mathcal{C}$ serve as an encoding of all such $\Psi$'s. 

\vgap

So far the number of $\Psi$'s is infinite, which does not fit the purpose of arguing for a space lower bound in RAM with a finite word length $w$. This can be easily fixed by creating $2^w$ choices for each vertex of $\Psi$, such that each of the $\eta$ vertices can independently take a choice of its own. This generates $2^{\eta w}$ polygons for $\Psi$, and hence, the same number of inputs in $\mathcal{C}$. Lemmas~\ref{lmm::lb-infer-vertices} and \ref{lmm::lb-dis-coeff} are still valid. Therefore, any encoding, which encodes a $(k/n,1/2)$-summary for each range space in $\mathcal{C}$, can be used to distinguish all those $2^{\eta w}$ choices of $\Psi$. The encoding, therefore, must use $\eta \cdot w$ bits for at least one range space. This completes the proof of Theorem~\ref{thm::lb}.


\pagebreak


\appendix 

\section*{Appendix} 
\section{Proof of Lemma~\ref{lmm::summary-rel_approx}} \label{app::lmm_rel_approx} 

\noindent {\bf Proof of the First Statement.} We will prove that, if a range $h \in \R_{i-1}$ satisfies $\ol{X}(h) \ge (1+\eps)/2^i$, then $\ol{S_i}(h) \cdot |X_i| + m_i \ge n/2^i$;  hence, $h$ cannot belong to $\R_i$. 

\vgap 


Consider first the case where $\ol{X_i}(h) < \rho_i$. Since $S_i$ is a $(\rho_i, \eps/4)$-approximation of $X_i$, we know
\begin{eqnarray}
    \ol{S_i}(h) &\ge& \ol{X_i}(h) - \rho_i \cdot \eps/4 = (\ol{X}(h) \cdot n - m_i)/|X_i| - \rho_i \cdot \eps/4 \nn \\
    &\ge& 
    \left(\fr{n(1+\eps)}{2^i} - m_i \right) \fr{1}{|X_i|} - \fr{n(1+\eps) \cdot \eps/4}{2^i \cdot |X_i|}  \nn  \\
    &=& 
    \fr{n(1+\eps)(1-\eps/4)}{2^i \cdot |X_i|} - \fr{m_i}{|X_i|} > 
    \fr{n}{2^i \cdot |X_i|} - \fr{m_i}{|X_i|} \nn  
\end{eqnarray}
as desired. Similarly, for the other case where $\ol{X_i}(h) \ge \rho_i$:
\begin{eqnarray}
    \ol{S_i}(h) &\ge& \ol{X_i}(h) (1 - \eps/4) = (\ol{X}(h) \cdot n - m_i)(1 - \eps/4)/|X_i|\nn \\
    &\ge& 
    \left(\fr{n(1+\eps)}{2^i} - m_i \right) \fr{1-\eps/4}{|X_i|}   \nn  \\
    &>& 
    \fr{n(1+\eps)(1-\eps/4)}{2^i \cdot |X_i|} - \fr{m_i}{|X_i|} 
    > 
    \fr{n}{2^i \cdot |X_i|} - \fr{m_i}{|X_i|}. \nn  
\end{eqnarray}

\extraspacing {\bf Proof of the Second Statement.} We will prove that every range $h \in \R_{i-1} \setminus \R_i$ satisfies $\ol{X}(h) \ge (1-\eps) / 2^i$. Combining this with the fact that $\R_i \subseteq \R_{i-1}$ will establish the second bullet.

\vgap

Given an $h \in \R_{i-1} \setminus \R_i$ but  $\ol{X}(h) < (1-\eps) / 2^i$, we will show that $\ol{S_i}(h) \cdot |X_i| + m_i < n/2^i$. Consider first the case where $\ol{X_i}(h) < \rho_i$. As $S_i$ is a $(\rho_i, \eps/4)$-approximation of $X_i$, we know
\begin{eqnarray}
    \ol{S_i}(h) &\le& \ol{X_i}(h) + \rho_i \cdot \eps/4 = (\ol{X}(h) \cdot n - m_i)/|X_i| + \rho_i \cdot \eps/4 \nn \\
    &\le& 
    \left(\fr{n(1-\eps)}{2^i} - m_i \right) \fr{1}{|X_i|} + \fr{n(1+\eps) \cdot \eps/4}{2^i \cdot |X_i|}  \nn  \\
    &=& 
    \fr{n(1 - \eps + \eps/4 + \eps^2/4)}{2^i \cdot |X_i|} - \fr{m_i}{|X_i|} 
    < 
    \fr{n}{2^i \cdot |X_i|} - \fr{m_i}{|X_i|} \nn  
\end{eqnarray}
as desired. Similarly, for the other case where $\ol{X_i}(h) \ge \rho_i$: 
\begin{eqnarray}
    \ol{S_i}(h) &\le& \ol{X_i}(h) (1 + \eps/4) = (\ol{X}(h) \cdot n - m_i)(1 + \eps/4)/|X_i|\nn \\
    &\le& 
    \left(\fr{n(1-\eps)}{2^i} - m_i \right) \fr{1 + \eps/4}{|X_i|}   \nn  \\
    &=& 
    \fr{n(1 - \eps)(1 + \eps/4)}{2^i \cdot |X_i|} - \fr{m_i}{|X_i|} 
    < 
    \fr{n}{2^i \cdot |X_i|} - \fr{m_i}{|X_i|}. \nn  
\end{eqnarray}

\section{Correctness of the Estimation Algorithm in Section~\ref{sec::summary-algo-2_est}} \label{app::qry_correct} 

\begin{proposition} \label{prop::proof-qry_correct-1}
    For any $h \in \R$ with $\ol{X}(h) \ge \rho$, our estimate $\ol{S_j}(h) \fr{|X_j|}{n} + \fr{m_j}{n}$ falls between $\ol{X}(h)(1-\eps)$ and $\ol{X}(h)(1+\eps)$. 
\end{proposition}
 
\begin{proof} 
    It is equivalent to show: 
    \begin{eqnarray}
        \left|\ol{S_j}(h) - \fr{\ol{X}(h) \cdot n - m_j}{|X_j|}\right|
        &\le&  \fr{n \cdot \ol{X}(h) \cdot \eps}{|X_j|} \nn \\
        \Leftrightarrow 
        \left|\ol{S_j}(h) - \ol{X_j}(h)\right|
        &\le&  \fr{n \cdot \ol{X}(h) \cdot \eps}{|X_j|} \label{app::qry_correct-1} 
    \end{eqnarray}
    Consider first the case where $\ol{X_j}(h) < \rho_j$. Since $S_j$ is a $(\rho_j, \eps/4)$-approximation of $X_j$, we know 
    \begin{eqnarray}
        |\ol{S_j}(h) - \ol{X_j}(h)| \le  \rho_j \cdot \eps/4 =  \fr{n(1+\eps) \cdot \eps/4}{2^j \cdot |X_j|} \label{app::qry_correct-2} 
    \end{eqnarray}
    By comparing \eqref{app::qry_correct-1} and \eqref{app::qry_correct-2}, we can see that it suffices to prove
    \begin{eqnarray} 
        \ol{X}(h) &\ge& \fr{1+\eps}{4 \cdot 2^j}. \label{app::qry_correct-3}
    \end{eqnarray}
    For this purpose, we distinguish two cases: 
    \begin{itemize} 
        \item[$-$]If $j = t$, then $\ol{X}(h) \ge \rho > 1/2^{t-1}$ (by definition of $t$), giving $\ol{X}(h) > (1+\eps)/(4 \cdot 2^{t})$.
        \item[$-$]Otherwise, by Lemma~\ref{lmm::summary-rel_approx}, we know $\ol{X}(h) \ge (1-\eps)/2^{j-1}$ which is at least $(1+\eps)/(4 \cdot 2^j)$ for $\eps \le 1/3$.
    \end{itemize}
    
    For the other case where $\ol{X_j}(h) \ge \rho_j$, we know
    \begin{eqnarray}
        |\ol{S_j}(h) - \ol{X_j}(h)| &\le&  \ol{X_j}(h) \cdot \eps/4  \label{app::qry_correct-4} 
    \end{eqnarray}
    By comparing \eqref{app::qry_correct-1} and \eqref{app::qry_correct-4}, we can see that it suffices to prove
    \begin{eqnarray} 
        \ol{X_j}(h) \cdot |X_j| &\le& 4 n \cdot \ol{X}(h) \nn
    \end{eqnarray}
    which is obviously true because $\ol{X_j}(h) \cdot |X_j| = |h \cap X_j| \le |h \cap X| = n \cdot \ol{X}(h)$.
\end{proof}

\begin{proposition} 
    For any $h \in \R$ with $\ol{X}(h) < \rho$, our estimate $\ol{S_j}(h) \fr{|X_j|}{n} + \fr{m_j}{n}$ falls between $\ol{X}(h)-\eps\rho$ and $\ol{X}(h)+\eps\rho$. 
\end{proposition}

\begin{proof} 
    Similar to the proof of the previous proposition, it is equivalent to show: 
    \begin{eqnarray}
        \left|\ol{S_j}(h) - \ol{X_j}(h)\right|
        &\le&  \fr{\eps\rho \cdot n}{|X_j|} \label{app::qry_correct-5} 
    \end{eqnarray}
    Consider first the case where $\ol{X_j}(h) < \rho_j$. Since $S_j$ is a $(\rho_j, \eps/4)$-approximation of $X_j$, we know 
    \begin{eqnarray}
        |\ol{S_j}(h) - \ol{X_j}(h)| \le  \rho_j \cdot \eps/4 =  \fr{n(1+\eps) \cdot \eps/4}{2^j \cdot |X_j|} \label{app::qry_correct-6} 
    \end{eqnarray}
    By comparing \eqref{app::qry_correct-5} and \eqref{app::qry_correct-6}, we can see that it suffices to prove
    \begin{eqnarray} 
        \rho &\ge& \fr{1+\eps}{4 \cdot 2^j}. \nn
    \end{eqnarray}
    The argument for \eqref{app::qry_correct-3} in the proof of Proposition~\ref{prop::proof-qry_correct-1} still holds here. Thus, the above follows from $\rho > \ol{X}(h)$. 
    
    \vgap
    
    For the other case where $\ol{X_j}(h) \ge \rho_j$, we know
    \begin{eqnarray}
        |\ol{S_j}(h) - \ol{X_j}(h)| &\le&  \ol{X_j}(h) \cdot \eps/4  \label{app::qry_correct-7} 
    \end{eqnarray}
    By comparing \eqref{app::qry_correct-5} and \eqref{app::qry_correct-7}, we can see that it suffices to prove
    \begin{eqnarray} 
        \ol{X_j}(h) \cdot |X_j| &\le& 4 n \rho \nn
    \end{eqnarray}
    which is true because $\ol{X_j}(h) \cdot |X_j| = |h \cap X_j| \le |h \cap X| = n \cdot \ol{X}(h) < n \rho$.
\end{proof}

%
%
%

\section{Proof of Lemma~\ref{lmm::halfspace-dist_coeff-box}} \label{app::lmm_proof-of-halfspace-dis_coeff-box} 

For any $0 < r < 1$, define $I(r)$ as the intersection between the region $\{x \in \real^d \mid \prod_{i=1}^d x[i] \le r\}$ and the unit box.
We will start by proving a geometric fact:

\begin{proposition} \label{prop::proof-dis_coeff-box-helping-1}
    For any $0 < r < 1$, $I(r)$ has a volume bounded by $O(r \cdot \log^{d-1} \fr{1}{r})$. 
\end{proposition}

\begin{proof} 
We partition $I(r)$ into two parts:

\begin{enumerate}
    \item $[r,1]^d \cap \{x \in \real^d \mid \prod_{i=1}^dx[i]\leq r\}$
    \item $\left([0,1]^d \setminus [r,1]^d \right) \cap \{x \in \real^d \mid \prod_{i=1}^dx[i]\leq r\}$
\end{enumerate}

We will prove that both parts have volume $O(r \cdot \log^{d-1} \fr{1}{r})$. In fact, this is obvious for the second part, because its volume is bounded by the volume of $[0,1]^d \backslash [r,1]^d$, which is $1^d - (1-r)^d=O(r)$.

\vgap 

The volume of the first part can be calculated explicitly as
\begin{eqnarray} 
    &&\int_{r}^1...\int_{r}^1 \left(\int_r^{r/\prod_{i=1}^{d-1} x[i]} {\bf{d}}x[d]\right){\bf{d}}x[d-1]...{\bf{d}}x[1] \nn \\
    &=&\int_{r}^1...\int_{r}^1 \left(r/{\prod_{i=1}^{d-1}x[i]}-r \right){\bf{d}}x[d-1]...{\bf{d}}x[1] \nn \\
    &\leq& r\int_{r}^1...\int_{r}^1 \frac{1}{\prod_{i=1}^{d-1}x[i]}{\bf{d}}x[d-1]...{\bf{d}}x[1] \nn \\
    &=& r \int_{r}^1...\int_{r}^1 \left(\frac{1}{\prod_{i=1}^{d-2}x[i]}\int_{r}^1\frac{1}{x[d-1]}{\bf{d}}x[d-1] \right){\bf{d}}x[d-2]...{\bf{d}}x[1] \nn \\
    &=& r \int_{r}^1...\int_{r}^1 \left(\frac{1}{\prod_{i=1}^{d-2}x[i]} \cdot \ln \fr{1}{r} \right){\bf{d}}x[d-2]...{\bf{d}}x[1] \nn \\
    &=& r \ln \fr{1}{r} \int_{r}^1...\int_{r}^1 \frac{1}{\prod_{i=1}^{d-2}x[i]} {\bf{d}}x[d-2]...{\bf{d}}x[1] \nn \\
    &...& \nn \\
    &=& r\ln^{d-1}\frac{1}{r}. \nn
\end{eqnarray}
\end{proof}

Let $p$ be one of the $2^d$ corners of the unit box. We say that a halfspace $h \in \R$ is {\em monotone} with respect to $p$ if $h$ satisfies the following conditions: (i) $h$ has a non-empty intersection with the unit box, and (ii) for every point $x$ covered by $h$ in the unit box, $h$ must contain the entire axis-parallel rectangle with $p$ and $x$ as opposite corners. The fact below is rudimentary: 

\begin{proposition} \label{prop::proof-dis_coeff-box-helping-2}
    If $h$ has a non-empty intersection with the unit box, then $h$ must be monotone to at least one corner of the unit box. \qed
\end{proposition}

Define $\R_1$ to be the subset of the halfspaces in $\R$ that are monotone to the origin of $\real^d$. With respect to the uniform distribution $\U$ inside the unit box, we have:

\begin{proposition} \label{prop::proof-dis_coeff-box-helping-3}
    For any $0 < r < 1$, 
    \begin{eqnarray}
        [0,1]^d \bigcap \left(\bigcup_{h \in \R_1 : \Pr_\U(h) \le r} h\right) \label{eqn::proof-prop-1} 
    \end{eqnarray}
    is a subset of $I(r)$.
\end{proposition}

\begin{proof}
    We will prove that if a point $x \in [0, 1]^d$ is outside $I(r)$, it does not belong to \eqref{eqn::proof-prop-1}. As $x \notin I(r)$, $x$ satisfies $\prod_{i=1}^d x[i] > r$. By monotonicity, any $h \in \R_1$ containing $x$ must cover the entire axis-parallel rectangle that has $x$ and the origin as opposite corners. Since this rectangle has volume $\prod_{i=1}^dx[i]$, we know $\Pr_\U(h) \geq \prod_{i=1}^dx[i] > r$. Hence, if $\Pr_\U(h) \le r$, $h$ must not contain $x$. It thus follows that $x$ is not in  \eqref{eqn::proof-prop-1}.
\end{proof}

We are now ready to prove Lemma~\ref{lmm::halfspace-dist_coeff-box}. Consider any halfspace $h^*$ that is disjoint with the unit box. We need to show that $\theta_\U^{h^*}(\sigma) = O(\log^{d-1} \fr{1}{\sigma})$ for all $\sigma > 0$. 

\vgap

For each $i \in [1, 2^d]$, define $\R_i$ as the subset of the halfspaces in $\R$ that are monotone to the $i$-th corner of the unit box (with the 1st corner being the origin). For any halfspace $h$ with a non-empty intersection with the unit box, $\Pr_\U[\dis(\{h^*, h\})] = \Pr_\U(h)$. 
Hence, $B_\U(h^*, r) = \{h \in \R \mid 0 \le \Pr_\U(h) \le r\}$. This further indicates that 
\begin{eqnarray} 
    [0,1]^d \cap \dis(B_\U(h^*, r))
    &=& [0,1]^d \cap \dis(\{h \in \R \mid 0 \le \Pr_\U(h) \le r\}) \nn \\
    &=& 
    [0,1]^d \cap \left(\bigcup_{h \in \R \hspace{1mm} : \hspace{1mm} 0 < \Pr_\U(h) \le r} h \right). \nn
\end{eqnarray}
By \eqref{eqn::coeff-dist_coeff}, we can compute $\theta_\U^{h^*}(\sigma)$ as
\begin{eqnarray} 
    \theta_\U^{h^*}(\sigma)
    &=& 
    \max\left\{1, \sup_{r > \sigma} \fr{\Pr_\U \left[[0,1]^d \cap \left(\bigcup_{h \in \R \hspace{1mm} : \hspace{1mm} 0 < \Pr_\U(h) \le r} h \right)\right]}{r} \right\}
    \nn \\ 
    \textrm{(by Proposition~\ref{prop::proof-dis_coeff-box-helping-2})} &\le& 
    \max\left\{1, \sum_{i=1}^{2^d} \sup_{r > \sigma} \fr{\Pr_\U \left[[0,1]^d \cap \left(\bigcup_{h \in \R_i \hspace{1mm} : \hspace{1mm} \Pr_\U(h) \le r} h \right)\right]}{r} \right\}
    \nn  \\ 
    \textrm{(by symmetry)}
    &=& 
    \max\left\{1, 2^d \cdot \sup_{r > \sigma} \fr{\Pr_\U \left[[0,1]^d \cap \left(\bigcup_{h \in \R_1 \hspace{1mm} : \hspace{1mm} \Pr_\U(h) \le r} h \right)\right]}{r} \right\}
    \nn  \\
    \textrm{(by Proposition~\ref{prop::proof-dis_coeff-box-helping-3})}&\le& 
    \max\left\{1, 2^d \cdot \sup_{r > \sigma} \fr{\Pr_\U [I(r)]}{r} \right\} \nn \\
    \textrm{(by Proposition~\ref{prop::proof-dis_coeff-box-helping-1})}&\le& 
    \max\left\{1, O\left(\sup_{r > \sigma} \fr{r \log^{d-1} \fr{1}{r}}{r} \right) \right\} 
    = 
    O\left(\log^{d-1} \fr{1}{\sigma}\right).
    \nn  
\end{eqnarray}

\section{Proof of Lemma~\ref{lmm::halfspace-dist_coeff-ball}} \label{app::lmm_proof-of-halfspace-dis_coeff-ball} 

For any $0 < r < 1$, define $b(r)$ as the ball $\{x \in \real^d \mid \sum_{i=1}^d x[i]^2 \leq r^2\}$; notice that $b(1)$ is simply the unit ball. We will establish the following fact at the end of the proof:

\begin{proposition} \label{prop::proof-dis_coeff-ball-helping-2}
    For any $0 < r < 1$, the volume of 
    \begin{eqnarray}
        b(1) \bigcap \left(\bigcup_{h \in \R \hspace{1mm} : \hspace{1mm} 0 < \Pr_\U(h) \le r} h\right) \label{eqn::proof-halfspace-dis_coeff-ball-helping-1} 
    \end{eqnarray} 
    is $O(r^\fr{2}{d+1})$.
\end{proposition}

Consider any halfspace $h^*$ disjoint with $b(1)$. We need to show that $\theta_\U^{h^*}(\sigma) = O((\frac{1}{\sigma})^\frac{d-1}{d+1})$ for all $\sigma > 0$. For any halfspace $h$ with a non-empty intersection with $b(1)$, $\Pr_\U[\dis(\{h^*, h\})] = \Pr_\U(h)$. 
Hence, $B_\U(h^*, r) = \{h \in \R \mid 0 \le \Pr_\U(h) \le r\}$, leading to 
\begin{eqnarray}
    b(1) \cap \dis(B_\U(h^*, r))
    &=&
    b(1) \cap \dis(\{h \in \R \mid 0 \le \Pr_\U(h) \le r\}) 
    \nn \\
    &=&
    b(1) \cap \left(\bigcup_{h \in \R \hspace{1mm} : \hspace{1mm} 0 < \Pr_\U(h) \le r} h \right). \nn
\end{eqnarray}
By \eqref{eqn::coeff-dist_coeff}, $\theta_\U^{h^*}(\sigma)$ can now be derived as
\begin{eqnarray}
    \theta_\U^{h^*}(\sigma)
    &=&
    \max\left\{1, \sup_{r > \sigma} \fr{\Pr_\U \left[b(1) \cap \left(\bigcup_{h \in \R \hspace{1mm} : \hspace{1mm} 0 < \Pr_\U(h) \le r} h \right)\right]}{r} \right\}
    \nn \\
    \textrm{(by Proposition~\ref{prop::proof-dis_coeff-ball-helping-2})} 
    &=&
    \max\left\{1, O\left(\sup_{r > \sigma} \fr{r^\fr{2}{d+1}}{r} \right) \right\}
    =
    O\left(\left(\fr{1}{\sigma}\right)^\frac{d-1}{d+1}\right).
    \nn
\end{eqnarray}

\begin{figure} 
    \centering
      \includegraphics[height=50mm]{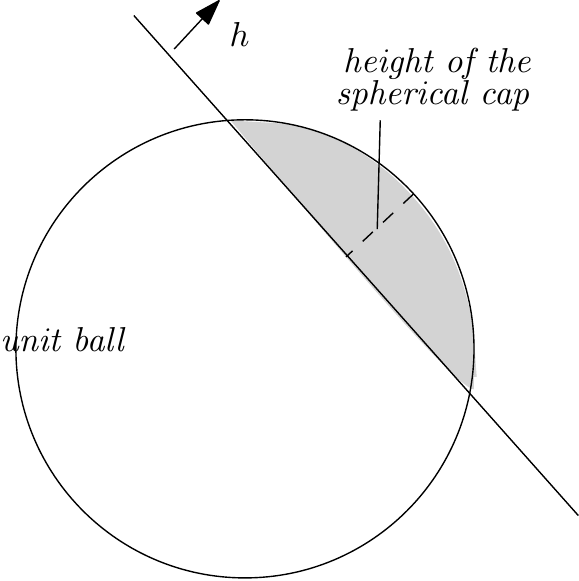}
  \figcapup 
  \caption{A 2D spherical cap} 
  \label{fig::app-cap}
  \figcapdown 
\end{figure}  

\extraspacing {\bf Proof of Proposition \ref{prop::proof-dis_coeff-ball-helping-2}.} Given a halfspace $h$, define $\EuScript{C}(h)$ as the ``spherical cap'' which is the intersection between $h$ and the unit ball. Furthermore, define the {\em height} of $\EuScript{C}(h)$ as the maximum Euclidean distance between a point in $\EuScript{C}(h)$ and the boundary plane of $h$. See Figure~\ref{fig::app-cap} for an illustration.

\begin{proposition} \label{prop::proof-dis_coeff-ball-helping-1}
    For any halfspace $h$, the volume of $\EuScript{C}(h)$ is $\Omega(r^\frac{d+1}{2})$, where $r$ is the height of $\EuScript{C}(h)$.
\end{proposition}
\begin{proof}
    Due to symmetry, we will prove the statement only for the halfspace $h = \{x \in \real^d \mid x[d] \ge 1-r\}$. Note that $\EuScript{C}(h)$ has height $r$. 
    
    \vgap

    Recall that, given a $(d-1)$-dimensional ball $b'$ in $\real^d$ (called the ``base'') and a point $p \in \real^d$ (called the ``apex vertex''), a  cone in $\real^d$ is defined as the union of the line segments connecting $p$ to all the points in $b'$. Consider the cone whose (i) base is the $(d-1)$-dimensional ball corresponding to the intersection of $b(1)$ and the plane $\{x \in \real^d \mid x[d] = 1 - r\}$, and (ii) its apex vertex is $(0,...,0,1)$. The cone has a volume $\Theta(r^\frac{d+1}{2})$. As the cone is contained by $\EuScript{C}(r)$, the volume of $\EuScript{C}(r)$ is $\Omega(r^\frac{d+1}{2})$.
\end{proof}

Let $h$ be a halfspace satisfying $0 < \Pr_\U(h) \le r$. The spherical cap $\EuScript{C}(h)$ has a volume that equals $\Pr_\U(h)$ multiplied by the volume of the unit ball. By Proposition~\ref{prop::proof-dis_coeff-ball-helping-1}, the height of $\EuScript{C}(h)$ is at most $c \cdot \Pr_\U(h)^\fr{2}{d+1}$ for some constant $c > 0$. Hence, every point in $h$ has a distance at least $1-c \cdot \Pr_\U(h)^\fr{2}{d+1} \geq 1-c \cdot r^\fr{2}{d+1}$ from the origin. 

\vgap

We now know that $\bigcup_{h : 0 < \Pr_\U(h) \le r} h$ --- the region given in \eqref{eqn::proof-halfspace-dis_coeff-ball-helping-1} --- is fully contained in the $d$-dimensional annulus $b(1) \setminus b(1-c \cdot r^\fr{2}{d+1})$. The annulus has volume $O(r^\fr{2}{d+1})$, thus completing the proof of Proposition~\ref{prop::proof-dis_coeff-ball-helping-2}.

\section{Proof of Lemma~\ref{lmm::lb-dis-coeff}} \label{app::lmm_proof-of-dis-coeff}

Let $\Psi$ be the polygon used to generate $X$. 
By Definition~\ref{eqn::coeff-sigma_min}, to compute $\theta_{X}(k/n)$, we should focus on the halfplanes $h \in \R$ with $\ol{X}(h) \le k/n$. Such halfplanes belong to one of the following categories: 
\begin{enumerate} 
    \item $h$ covers the $k$ outer points, but not any inner points of $X$.
    \item $h$ does not cover any points of $X$ at all. 
\end{enumerate}

\noindent {\bf Category 1.}
Consider an $h$ of the first category. By \eqref{eqn::coeff-set_coeff_range}, to derive $\theta^h_{\U(X)}(k/n)$, we need to calculate the ratio $\fr{1}{r} \cdot \Pr_{\U(X)}[\dis(B_{\U(X)}(h,r))]$ at each $r > k/n$. 
\begin{itemize} 
    \item[$-$]Case 1: $r \ge \fr{q/\eta}{n}$. Let $h'$ be any upward halfplane covering exactly one vertex $v$ of $\Psi$. $\dis(\{h, h'\}) \cap X$ is the set of $q/\eta$ inner points at $v$. Therefore, $h' \in B_{\U(X)}(h, r)$, which indicates that $\dis(B_{\U(X)}(h, r))$ covers all the inner points of $X$. 
    
    \vgap
    
    \hspace{4mm} Consider now a halfplane $h'' \in \R$ that does not cover any points of $X$. $\dis(\{h, h''\}) \cap X$ is the set of $k$ outer points. Since $k < q/\eta$, $h'' \in B_{\U(X)}(h, r)$, which indicates that $\dis(B_{\U(X)}(h, r))$ covers all the outer points of $X$ as well. 
    
    \vgap
    
    \hspace{4mm} Therefore, $\Pr_{\U(X)}[\dis(B_{\U(X)}(h,r))] = 1$. This means that $\fr{1}{r} \cdot \Pr_{\U(X)}[\dis(B_{\U(X)}(h,r))]$ is maximized at $\fr{n}{q/\eta}$ for $r \ge \fr{q/\eta}{n}$. 
    
    \item[$-$]Case 2: $\fr{k}{n} < r < \fr{q/\eta}{n}$. Consider an arbitrary halfplane $h' \in B_{\U(X)}(h,r)$. Since $h'$ cannot cover any vertex of $\Psi$,\footnote{Otherwise, $h'$ contains at least $q/\eta$ inner points, giving $\Pr_{\U(X)}[\dis(\{h,h'\})] \ge \fr{q/\eta}{n}$ and contradicting $h' \in B_{\U(X)}(h,r)$.} $\dis(B_{\U(X)}(h, r))$ contains no inner points. On the other hand, $h'$ may be a halfplane that covers no points of $X$ at all. Hence, $\dis(B_{\U(X)}(h, r))$ is precisely the set of $k$ outer points. Therefore, the supremum of $\fr{1}{r} \cdot \Pr_{\U(X)}[\dis(B_{\U(X)}(h,r))]$ is $1$ when $\fr{k}{n} < r < \fr{q/\eta}{n}$. 
\end{itemize}
We conclude that $\theta^h_{\U(X)}(k/n) = \fr{n}{q/\eta}$ for any $h$ in Category 1.

\extraspacing {\bf Category 2.} Consider an $h$ of the secondary category. To compute $\theta^h_{\U(X)}(k/n)$, we distinguish three cases:
\begin{itemize} 
    \item[$-$]Case 1: $r \ge \fr{k}{n} + \fr{q/\eta}{n}$. Let $h'$ be any upward halfplane covering exactly one vertex $v$ of $\Psi$. $\dis(\{h, h'\}) \cap X$ is the set of $q/\eta$ inner points at $v$, plus the $k$ outer points. Therefore, $h' \in B_{\U(X)}(h, r)$. This indicates that $\dis(B_{\U(X)}(h, r))$ is the entire $X$. Therefore, $\Pr_{\U(X)}[\dis(B_{\U(X)}(h,r))] = 1$; and $\fr{1}{r} \cdot \Pr_{\U(X)}[\dis(B_{\U(X)}(h,r))]$ is maximized at $\fr{n}{k + q/\eta}$ for $r \ge \fr{k}{n} + \fr{q/\eta}{n}$. 
    
    \item[$-$]Case 2: $\fr{q/\eta}{n} \le r < \fr{k}{n} + \fr{q/\eta}{n}$. Consider an arbitrary halfplane $h' \in B_{\U(X)}(h,r)$. Since $h'$ may be a halfplane that covers all the $k$ outer points but no inner points, $\dis(B_{\U(X)}(h, r))$ contains all the $k$ outer points. 
    
    \vgap
    
    \hspace{4mm} 
    Note that $h'$ can cover at most one vertex of $\Psi$. Furthermore, if it covers a vertex of $\Psi$, it cannot contain any outer points. Recall that all the vertices of $\Psi$ are on the upper arc of the unit circle. Hence, if $h'$ covers a single vertex of $\Psi$, the vertex must either be the leftmost vertex $v_\mathit{left}$ of $\Psi$, or the rightmost one $v_\mathit{right}$.\footnote{Suppose that the (only) vertex $v$ of $\Psi$ covered by $h'$ is neither $v_\mathit{left}$ nor $v_\mathit{right}$. Then, the boundary of $h'$ must intersect both the segment connecting $v_\mathit{left}$, $v$, and the segment connecting $v_\mathit{right}$, $v$. This suggests that $h'$ must be upward, but this contradicts the fact that $h'$ covers no outer points.} This indicates that $\dis(B_{\U(X)}(h, r))$ contains $2q/\eta$ inner points: $q/\eta$ ones placed at $v_\mathit{left}$ and $v_\mathit{right}$, respectively. 
    
    \vgap
    
    \hspace{4mm} Therefore, $\Pr_{\U(X)}[\dis(B_{\U(X)}(h,r))] = \fr{k + 2q/\eta}{n}$. Hence, $\fr{1}{r} \cdot \Pr_{\U(X)}[\dis(B_{\U(X)}(h,r))]$ is maximized at $\fr{k + 2q/\eta}{q/\eta}$ when $\fr{q/\eta}{n} \le r < \fr{k}{n} + \fr{q/\eta}{n}$.
    
    \item[$-$]Case 3: $\fr{k}{n} < r < \fr{q/\eta}{n}$. Consider an arbitrary halfplane $h' \in B_{\U(X)}(h,r)$. Since $h'$ cannot cover any inner points, $\dis(B_{\U(X)}(h, r))$ contains no inner points. On the other hand, $h'$ may be a halfplane that covers all the $k$ outer points, indicating that $\dis(B_{\U(X)}(h, r))$ is the set of $k$ outer points. Therefore, the supremum of $\fr{1}{r} \cdot \Pr_{\U(X)}[\dis(B_{\U(X)}(h,r))]$ is $1$ when $\fr{k}{n} < r < \fr{q/\eta}{n}$. 
\end{itemize}
When $k, q$, and $\eta$ satisfy the requirements of Theorem~\ref{thm::lb}, $\fr{n}{k+q/\eta} > \fr{k+2q/\eta}{q/\eta} > 1$. We conclude that $\theta^h_{\U(X)}(k/n) = \fr{n}{k+q/\eta}$ for any $h$ in Category 2.

\vgap 

Combining the above discussion of both categories shows that $\theta_{X}(k/n) = \fr{n}{k+q/\eta}$, as claimed in the lemma.

%

\bibliographystyle{plain}
\bibliography{ref}

\end{sloppy}
\end{document}